\documentclass[12pt]{article}

\newcommand{\bC}{\mathbb {C}}

\newcommand{\cO}{\mathcal{O}}

\newcommand{\cH}{\mathcal{H}}

\newcommand{\bN}{\mathbf{N}}

\newcommand{\cQ}{\mathcal{Q}}

\newcommand{\bR}{\mathbf{R}}

\newcommand{\bz}{{\bar{z}}}

\newcommand{\ii}{\mathrm{i}}

\usepackage{enumerate}
\usepackage{amsmath}    
\usepackage{amsthm}     
\usepackage{amscd}      
\usepackage{amssymb}    

\usepackage{eucal}      
\usepackage{latexsym}   
\usepackage{graphicx}   
\usepackage{verbatim}   
\usepackage{gimmsy2e}   
\usepackage{epstopdf}
\usepackage{color}
\usepackage[colorlinks=true, pdfstartview=FitV, linkcolor=blue, citecolor=blue, urlcolor=blue]{hyperref}
\usepackage{stmaryrd}
\usepackage{etex}
\usepackage[utf8]{inputenc}
\usepackage[TS1,T1]{fontenc}
\usepackage{xspace}
\usepackage{lmodern}

\usepackage{pifont,calc,color,relsize}

\newcommand{\blackcircle}[1]{\makebox[0pt][l]{{\relsize{2.3}\raisebox{-0.25\totalheight}{\ding{108}}}}\makebox[\widthof{\relsize{2}\ding{108}}][c]{\textcolor{white}{#1}}}

\newcommand{\bmu}{\stackrel{{\tiny\blackcircle{$\boldsymbol{\mu}$}}}{\ }}
\newcommand{\bnu}{\stackrel{{\tiny\blackcircle{$\boldsymbol{\nu}$}}}{\ }}
\newcommand{\bpi}{\stackrel{{\tiny\blackcircle{$\boldsymbol{\pi}$}}}{\ }}

 \numberwithin{thm}{subsection}

\makeatletter
\@addtoreset{equation}{subsection}
\def\thesubsection{\thesection\Alph{subsection}}
\def\thesection{\arabic{section}}

\def\theequation{\thesubsection.\arabic{equation}}

\makeatother

\def\intprod{\mathbin{\hbox to 6pt{%
                 \vrule height0.4pt width5pt depth0pt
                 \kern-.4pt
                 \vrule height6pt width0.4pt depth0pt\hss}}}

\newcommand{\ns}{\mspace{-1.5mu}}             
\newcommand{\ps}{\mspace{1.5mu}}              

\newpage
\definecolor{gray}{rgb}{0.5, 0.5, 0.5}
\definecolor{gold}{rgb}{0.255, 0.215, 0}
\definecolor{cyan}{rgb}{0.0, 0.2, 0.2}
\definecolor{pink}{rgb}{0.2, 0.1, 0.2}

\begin{document}

\baselineskip=20pt

\title{Quantization via Deformation of Prequantization}

\author{
{\sc Christian Duval}
\\[-2pt]
Centre de Physique Th\'eorique, Case 907\\[-2pt]
13288 Marseille Cedex 9 France\\[-2pt]
(email: duval@cpt.univ-mrs.fr)\\[6pt]
\and 
{\sc Mark~J.~Gotay}
\\[-2pt]
Pacific Institute for the Mathematical Sciences\\[-2pt]
University of British Columbia\\[-2pt]
Vancouver, BC  V6T 1Z2  Canada\\[-2pt]
(email: gotay@pims.math.ca)\\[6pt]}

\date{\today} 

\maketitle
\begin{abstract}\footnotesize{We introduce the notion of a ``Souriau bracket'' on a prequantum circle bundle $Y$ over a phase space $X$ and  explain how a deformation of $Y$ in the direction of this bracket provides a \emph{genuine} quantization of $X$.}\\

\end{abstract}

\section{Introduction}

Let $(X,\omega)$ be a symplectic manifold representing the phase space of a classical system. Consider a prequantum $U(1)$-bundle $(Y,\alpha)$ over  $(X,\omega)$;  here $\alpha$ is a connection 1-form on $Y$ satisfying the curvature condition $d\alpha = \omega/\hbar$.\footnote{\, A number of identifications will be made throughout; in particular, pullbacks will be ruthlessly suppressed.}  Let $\pi$ be the Poisson bivector associated to $\omega$, and let $\pi^\#$ be the 
horizontal lift of $\pi$ to $Y$ with respect to $\alpha$.\footnote{\, That is, $\pi^{\#}$ is the unique $U(1)$-invariant bivector on $Y$ which vanishes on $\alpha$ and projects to $\pi$ on $X$.} Define the \emph{Souriau bracket} of $f,g \in C^\infty(Y,\mathbb C)$ to be\footnote{\, This is related to the Lagrange (or Jacobi) bracket on $Y$. See \cite{Ma1991}, Example~5 (\S2.3) and Example~2.5.} 
$$\llbracket f,g \rrbracket := \pi{}^{\#}(df,dg).$$
Our aim is to use the Souriau bracket together with a polarization to quantize $(X,\omega)$.

Suppose $\Psi:Y \to \mathbb C$  is a $U(1)$-equivariant map, that is to say, a \emph{prequantum wave function}. Then---and this is a key point---for a \emph{classical observable} $F \in C^\infty(X,\mathbb C)$ the Souriau bracket 
$\llbracket F,\Psi \rrbracket$ is readily verified to be $U(1)$-equivariant as well, so that
 \begin{equation}
{\mathcal PF[\Psi]} := F\Psi + \frac{\hbar}{\ii} \ps \llbracket F,\Psi  \rrbracket
\label{pqop}
\end{equation}
 is also a prequantum wave function. It turns out 
that  $\mathcal PF[\Psi]$ is none other than the prequantum operator corresponding to the classical observable $F$ acting on the prequantum wave function $\Psi$. 

\medskip

The relation \eqref{pqop} brings to mind the lowest order terms in the expression (in the `deformation parameter' $\hbar$) for a `star' product $F \bullet \Psi$ on $Y$ with driver the Souriau bracket. However, while the latter is readily verified to satisfy the Leibniz rule, it is not a Lie bracket as the Jacobi identity fails.\footnote{\, Indeed, the Schouten-Nijenhuis bracket $[\pi^{\#},\pi^{\#}] = -2 \pi^{\#} \wedge \eta$, where $\eta$  
is the Reeb vector field of $(Y,\alpha)$.} So such a $\bullet$ cannot be an associative deformation of $C^\infty(Y,\mathbb C)$ and hence not, strictly speaking, a star product.\footnote{\, A star product has died.}

Nonetheless, it is natural to wonder if  one can develop a formula 
`extending'  \eqref{pqop} to an appropriate ``quantum product'' $\bullet$ on $Y$, with driver the Souriau bracket, in such a way that 
 \begin{equation*}
\mathcal QF[\Psi]  :=   F \bullet \Psi 
\end{equation*}
gives a genuine \emph{quantum operator} $\mathcal QF$ corresponding to $F \in C^\infty(X,\mathbb C)$? It turns out that we can, at least under certain circumstances---thus in this sense we obtain quantization via a deformation of prequantization.  We show here how this works.

Further motivation for our approach stems from the observation that the two key classes of objects one must consider for quantization are observables (i.e., elements of $C^\infty(X,\mathbb C)$) and prequantum wave functions (i.e.,  $U(1)$-equivariant functions on $Y$). As  these naturally form subspaces  $\mathcal O$ and $\mathcal H$  of $C^\infty(Y,\mathbb C)$, respectively, it seems appropriate to take $Y$ as the arena for quantization.  In addition, the introduction of the quantum product $\bullet$  will `add functionality' to the geometric quantization scheme. (As is well known, the latter has difficulty quantizing sufficiently many observables. The introduction of a quantum product will mollify this, cf. Example 5C.)

\section{Coordinate Expressions}
\setcounter{equation}{0}

If $p_i,q^j$ are canonical coordinates on $(X,\omega)$, and $\theta$ is the (angular) fiber coordinate on $Y$, then locally
 \begin{equation*}
\omega = dp_i \wedge dq^i
\end{equation*}
and
 \begin{equation}
 \label{connection}
\alpha = \frac{1}{\hbar}\ps p_i\ps dq^i + d\theta.
\end{equation}
The Reeb vector field is thus $\eta = \partial_\theta$.
 The associated Poisson bivector $\pi = -\omega^{-1}$  is
\begin{equation*}
\label{pb}
\pi=\partial_{p_i}\wedge\partial_{q^i}
\end{equation*}
and so the \emph{Souriau bivector}  $\pi^{\#} = \partial_{p_i}{}^{\#}\wedge\partial_{q^i}{}^{\#}$ is
\begin{equation*}
\label{st}
\pi{}^{\#}=\partial_{p_i}\wedge\left(\partial_{q^i} - \frac{p_i}{\hbar}\partial_\theta\right),
\end{equation*}
whence 
\begin{equation}
\label{sb}\llbracket f,g\rrbracket = \frac{\partial f}{\partial p_j}\left(\frac{\partial g}{\partial q^j} - \frac{p_j}{\hbar}\frac{\partial g}{\partial \theta}\right) - \frac{\partial g}{\partial p_j}\left(\frac{\partial f}{\partial q^j} - \frac{p_j}{\hbar}\frac{\partial f}{\partial \theta}\right) .
\end{equation}

Using the local expression $\Psi = \psi(p,q)e^{\ii\theta}$ for the prequantum wave function and \eqref{pqop}, we obtain
 \begin{equation*}
\mathcal PF[\Psi] = \left(\frac{\partial F}{\partial p_j}\left(\frac{\hbar}{\ii}\frac{\partial \psi}{\partial q^j}-  p_j\ps  \psi  \right) - \frac{\hbar}{\ii}\frac{\partial F}{\partial q^j}\frac{\partial \psi}{\partial p_j}+ F\psi\right)\ns e^{\ii\theta},
\label{pqop1}
\end{equation*}
from which is evident that $\mathcal PF$ is indeed the prequantization of the classical observable $F$.

\section{The Quantization Construction}
\setcounter{equation}{0}

Our approach seeks to find a `middle ground' between geometric quantization and deformation quantization, utilizing the most successful aspects of both.\footnote{\ Useful references on geometric and deformation quantization are \cite{Sn1980,So1997}, and \cite{DS2002}, respectively.}
It is not surprising, then, that  the central elements of our quantization construction are prequantization,  a ``quantum product'' and a polarization. We have briefly discussed prequantization in the first section.

We assume  that we have a star-product $\star_\lambda$ on $X$, with deformation parameter $\lambda $.\footnote{\, Such always exist by virtue of \cite{DWLe1983}. We do not consider questions of convergence here.} The driver of this star product is a contravariant 2-tensor $\Lambda$ on $X$.
 The most important case is when $\Lambda$ is the Poisson bivector itself, but we will also consider other possibilities. 

In addition to  the star product we will need a deformation quantization of $\cH$ with respect to $\star_\lambda$ in the sense of \cite{BNWW2010}. 
By this we mean a one-parameter family of 
$\mathbb C [[\lambda]]$-bilinear $U$(1)-equivariant mappings
$$
\bullet_\lambda :\cO[[\lambda]]\times{}\cH[[\lambda]]\to{}\cH[[\lambda]]
$$
satisfying
\begin{equation}
\label{Module}
(F \star_\lambda G) \bullet_\lambda \Psi = F \bullet_\lambda(G\bullet_\lambda \Psi).
\end{equation}
We require further that the \emph{prequantum product} $\bullet_\lambda$ be of the form 
$$F\bullet_\lambda \Psi = \sum_{k=0}^\infty \lambda^k c_k(F,\Psi),$$
where each $c_k$ is bidifferential, $c_0(F,\Psi) =F\ps \Psi$ and the driver of $\bullet_\lambda$ is $c_1(F,\Psi) = \Lambda^\#(dF,d\Psi)$.
Theorem 1.6 of \cite{BNWW2010} implies that such a deformation quantization of $\cH$ exists and is unique up to equivalence.\footnote{\ This reference actually considered deformation quantizations of $C^\infty(Y,\mathbb C)$. Since the subspace $\cH$ thereof consists of $U$(1)-equivariant functions, a deformation quantization of  $C^\infty(Y,\mathbb C)$ automatically restricts to one of $\cH$. }

Finally, we suppose that $(X,\omega)$, with $\dim X = 2n$, is equipped with a polarization $J$. 
A prequantum wave function $\Psi$ is \emph{polarized} provided $\zeta^{\#}[\Psi] = 0$ for all $\zeta \in J$. The space of all polarized prequantum wave functions (or simply ``wave functions'' for short) is denoted $\mathcal H_J$.
We say that   the prequantum product is \emph{compatible} with the polarization if whenever $\Psi$ is polarized, then so is $F \bullet_\lambda \Psi$ for all $F \in \cO$. Then \eqref{Module} implies that
$\cH_J$ is a left $\big(\cO,\star_\lambda\big)$-module, and we call $\bullet_\lambda$ a \emph{quantum product}.

Thus far Planck's constant has not appeared in our formul\ae\ (except in the expression \eqref{sb} for the Souriau bracket). To turn our constructions into a quantization in a \emph{physical} sense we now insist that we evaluate the expansions above at $\ii \lambda = \hbar$, the numerical value of Planck's constant. When the driver of the star product is the Poisson bracket we then have 
\begin{equation}
\label{expansion}
F \bullet_\hbar \Psi = {F\, \Psi} + \frac{\hbar}{\ii}\, {\llbracket F,\Psi \rrbracket} + \cdots
\end{equation}
for each wave function $\Psi$, which is the raison d'\^etre for our construction of the quantum product.  We caution that \eqref{expansion} cannot \emph{a priori} be interpreted as a power series in $\hbar$, again because $c_1$ depends upon $\hbar$ and also since we do not in general have any control over the $\hbar$-dependence of the $c_k$ for $k \geq 2$.

Henceforth we drop the $\hbar$-dependence in the notation and simply refer to~$\bullet$ as \emph{the} quantum product. Now for $\Psi \in \cH_J$ define
\begin{equation*}
\cQ\in\mathrm{Hom}\big( (\cO,\star),\mathrm{End}(\cH_J,\circ)\big)
\label{Q}
\end{equation*}
by 
\begin{equation}
\label{qop}
\cQ F[\Psi] := F \bullet \Psi;
\end{equation}
  then   \eqref{Module} takes the familiar form
\begin{equation}
\label{convolve} 
\mathcal Q(F\star G)[\Psi] =( \mathcal Q F \circ \mathcal QG)[\Psi].
\end{equation}
Note also that $\cQ(1)=  1_{\cH_J}$. These properties suggest, and the examples in \S 5 will  justify,  identifying $\mathcal Q F$ with the \emph{genuine} quantum operator associated to the observable $F$ in the representation determined by $J$.

  We do not know a general way of constructing quantum products, nor do we know if every prequantum circle bundle (over a polarized symplectic manifold) carries a quantum product. Nonetheless, it is possible to construct such products, at least in a certain class of examples, as we illustrate in the next section.

\section{Quantum Products over Flat Bi-Polarized Manifolds}
\setcounter{equation}{0}

From now on we specialize to the case when $(X,\omega)$ is a flat bi-polarized manifold. By  `bi-polarized' we mean that $X$ carries two transverse real polarizations $J,K$.\footnote{\, This concept is the same as that of a `bi-Lagrangian' manifold' that one finds in the literature (cf. \cite{Bo1993,EST2006}, and references therein),
except that we allow for {\em complex} Lagrangian distributions. This has essentially no effect on what follows.} 
 The first one, $J$, is used to polarize prequantum wave functions as described previously. The second one, $K$, taken together with $J$, allows us to `polarize' the Poisson bivector, i.e., display $\pi$ in a certain normal form which we now describe.

Hess constructed a canonical symmetric symplectic connection $\nabla$ on such a space
and proved the following key result.

\begin{thm}[Hess \cite{He1980}]
 The following conditions are equivalent:
\begin{enumerate}
 \item[{\rm 1.}] $\nabla$ is  flat.
\item[{\rm 2.}] About every point in $X$ there exists a 
Darboux chart 
 $\{p_1,  \ldots , p_n, \linebreak q^1, \ldots, q^n\}$  such that the polarizations  $J$ and $K$ are 	locally	generated 	by	the	Hamiltonian	vector	fields	 $\xi_{q^1} , \ldots, \xi_{q^n}$ and	$\xi_{p_1} , \ldots, \xi_{p_n}$, respectively.
 \end{enumerate}
\end{thm}
 
As a consequence we have the local normal form
\begin{equation}
\pi = 
\partial_{p_j} \wedge \partial_{q^j},
\label{pinf}
\end{equation}
where the $ \partial_{p_j} \; (\; = - \xi_{q^j}) $ lie in $J$ and  the $ \partial_{q^j} \; (\; = \xi_{p_j}) $ lie in $K$. Furthermore it is straightforward to verify that any two such charts are affinely related: on overlaps 
\begin{equation}
\label{overlaps}
p'_j = a_j{}^i\ps p_i + b_j \qquad {\rm and } \qquad q'^j = c^j{}_i\ps q^i + d^j
\end{equation}
where  $a = c^{-1}$. Thus a flat bi-polarized manifold is affine. For more details, including numerous examples, see \cite{Bo1993,EST2006,Li1954}.

Now
suppose that we have a  contravariant 2-tensor $\Lambda$
which on an open set $U \subset X$  can be expanded 
\begin{equation}
\Lambda \ps | \ps U = \sum_\alpha s^ \alpha \otimes t_\alpha
\label{decomp}
\end{equation} 
into tensor products of mutually commuting vector fields $s^ \alpha,t_ \beta $ (which again may be complex-valued) and which transform contragrediently. This is the case in particular when $U$ is a Darboux chart and  $\Lambda$ is (\emph{i}) the Poisson bivector \eqref{pinf},  (\emph{ii}) the associated 
\emph{normal tensor} 
$$\nu = \partial_{p_k} \otimes \partial_{q^k}$$
 and  (\emph{iii}) the \emph{anti-normal} tensor 
$$\mu  =   -  \partial_{q^k} \otimes \partial_{p_k}.$$

We  regard $\Lambda  \ps | \ps U $ as a bilinear map $C^\infty(U,\mathbb C) \otimes C^\infty(U,\mathbb C) \to C^\infty(U,\mathbb C) \otimes C^\infty(U,\mathbb C)$ according to
\begin{equation}
\label{bmap}
(\Lambda  \ps | \ps U ) (F \otimes G) = \sum_ \alpha s^ \alpha[F] \otimes t_ \alpha[G]
\end{equation}
and we then put $(\Lambda  \ps | \ps U )^{k+1} = (\Lambda  \ps | \ps U ) \circ (\Lambda  \ps | \ps U )^{k}$.
Note that once we have fixed a decomposition \eqref{decomp} of $\Lambda$ on $U$ there are no factor-ordering ambiguities when computing $(\Lambda \ps | \ps U ) ^k(F\otimes G)$ for $k \geq 2$. 

Next, cover the bi-polarized manifold $X$ with an  atlas of affine Darboux charts with overlaps as in \eqref{overlaps}. Then if  $U,U'$ are two such domains, equation \eqref{bmap} and the assumption that the $s^\alpha,t_\beta$ transform contragrediently imply that 
$(\Lambda  \ps | \ps U )^k (F \otimes G) = (\Lambda  \ps | \ps U' )^k (F \otimes G)$ on $U\cap U'$.
Thus $\Lambda^k(F\otimes G)$ is \emph{globally} well-defined. 

These are the first crucial consequences of the normal form of Hess' theorem.

A calculation shows that we may use such a $\Lambda$ to drive a (formal) star product $\star$ on $C^\infty(X,\mathbb C)$ of `exponential type,'  i.e., of the form
\begin{equation*}
\label{starprod}
F \star G =   \sum_{k=0}^\infty \left(\frac{\hbar}{\ii}\right)^k \frac{1}{k!}\ m \circ\Lambda^k(F \otimes G)
\end{equation*}
for $F,G \in C^\infty(X,\mathbb C)$
where $m$ is the multiplication operator $m(F \otimes G) = FG$.  For short, we write
$$F\star G = m \circ \exp\left(\frac{\hbar}{\ii}\ps \Lambda\right)(F\otimes G).$$

We horizontally lift this star product to induce a  product on $C^\infty(Y,\mathbb C)$: for $f,g \in C^\infty(Y,\mathbb C)$ we set
\begin{equation}
f \bullet g =  m \circ \exp\left(\frac{\hbar}{\ii}\ps \Lambda^\#\right) (f \otimes g).
\label{qp}
\end{equation}
Since
\begin{equation}
\label{liftoflambda}
\Lambda^\# \ps | \ps U = \sum_\alpha (s^ \alpha)^\# \otimes (t_\alpha)^\#,
\end{equation} 
the product $\bullet$ is also globally well-defined. As one would expect $F\bullet G = $ \linebreak[4] $= F \star G$  for $F,G \in \cO$.
Furthermore
 \begin{prop} 
 $\big(C^\infty(Y,\mathbb C),\bullet\big)$ is a left $(\cO,\star)$-module, i.e.,
 \begin{equation}
 (F \star G) \bullet h = F\bullet (G\bullet h)
 \label{module}
 \end{equation}
 for all $F,G \in \cO$ and $h \in C^\infty(Y, \mathbb C).$
  \label{BundleDQ}
\end{prop}
\begin{proof}
Expand both sides of \eqref{module} using \eqref{qp} and the multinomial formula. Since the $s^\alpha, t_\beta$ mutually commute, we may then rearrange and  reindex one side to obtain the other.
\end{proof}

By virtue of its construction in terms of horizontal lifts, $\bullet$ is $U$(1)-equivar\-iant. Thus it is a prequantum product, and so provides a deformation quantization of $C^\infty(Y, \mathbb C).$ 

We next prove that $\bnu$ is a genuine quantum product, $\bnu$ being given by \eqref{qp} with $\Lambda$   the normal tensor $\nu$. This is the second crucial consequence of Hess' theorem.

\begin{prop} If $H$ is an observable and $\Psi$ is $J$-polarized, $H \bnu \Psi$ is also $J$-polarized. 
 \label{polprop}
\end{prop}
\begin{proof}
It suffices to show that for any $\ell$ we have   $\partial_{p_\ell}{}^\#[H \bnu \Psi] = 0.$
A multinomial expansion of \eqref{qp} gives   
\begin{align*}
H \bnu \Psi = \sum_{k\in \mathbb N}\left(\frac{\hbar}{\ii}\right)^k &  \frac{ 1 }{k!}   \sum_{j_1+ \cdots + j_n=k} 
\left(\begin{array}{cc}
k \\
j_1 \cdots j_n
\end{array}
\right) \\[2ex]
& \times \Big( (\partial_{p_{1}})^{j_1} \cdots (\partial_{p_{n}})^{j_n} [H] 
  \times ( \partial_{q^{1}}{}^{\#})^ {j_1} \cdots ( \partial_{q^{n}}{}^{\#})^ {j_n}[\Psi]\Big),
\end{align*}
where we have used the fact that the vector fields $\partial_{q^m}{}^\#$ commute (as $K$ is a polarization). Then
\begin{align}
\partial_{p_\ell}{}^\#[H \bnu \Psi] =  & \sum_{k \in \mathbb N} \left(\frac{\hbar}{\ii}\right)^k \frac{1}{k!}  \sum_{j_1+ \cdots +j_n=k}  \left(\begin{array}{cc}
k \\
j_1 \cdots j_n
\end{array}
\right)
\label{hairy eq} \\[2ex]
&\  \times \Big( \partial_{p_\ell} (\partial_{p_{1}})^{j_1} \cdots (\partial_{p_{n}})^{j_n} [H] 
 \times ( \partial_{q^{1}}{}^{\#})^ {j_1} \cdots ( \partial_{q^{n}}{}^{\#})^ {j_n}[\Psi] \nonumber \\[2ex]
& \ \ \  \ + (\partial_{p_{1}})^{j_1} \cdots (\partial_{p_{n}})^{j_n} [H] 
  \times\partial_{p_\ell}{}^\#( \partial_{q^{1}}{}^{\#})^ {j_1} \cdots ( \partial_{q^{n}}{}^{\#})^ {j_n}[\Psi]
\Big). \nonumber
\end{align}

Before proceeding, we note the following facts. By the definition of curvature, the prequantization condition ${\rm curv\ } \alpha = \omega/\hbar$ and \eqref{connection},
\begin{eqnarray*}
[\partial_{p_\ell}{}^\#, \partial_{q^m}{}^\# ]   & = & [\partial_{p_\ell}{}, \partial_{q^m}]^\# - \frac{1}{\hbar}\omega(\partial_{p_\ell}{}, \partial_{q^m})\eta \\[2ex]
& = & 0 -\frac{1 }{\hbar}\delta^\ell{}_m\ps \eta \end{eqnarray*}
which follows from \eqref{pinf}. Similarly, we compute that
$ [\eta, \partial_{q^m}{}^\# ]  = 0.$

Now, we manipulate the second factor of the last term in the sum \eqref{hairy eq}. Consider the quantity
\begin{equation}
\label{commutator}
\partial_{p_\ell}{}^\# (\partial_{q^m}{}^\#)^{j} = (\partial_{q^m}{}^\#)^j\partial_{p_\ell}{}^\#  \\
 +  [\partial_{p_\ell}{}^\#, (\partial_{q^m}{}^\#)^j ].
\end{equation}
In the second term here, expand
\begin{eqnarray*}
 [\partial_{p_\ell}{}^\#, (\partial_{q^m}{}^\#)^j ]  =   \partial_{q^m}{}^\#[ \partial_{p_\ell}{}^\# , (\partial_{q^m}{}^\#)^{j-1} ] \nonumber 
 +   [\partial_{p_\ell}{}^\#, \partial_{q^m}{}^\# ] (\partial_{q^m}{}^\#)^{j-1}.
\end{eqnarray*}
Iterating this last computation $(j-1)$-times and taking into account the facts listed above, equation \eqref{commutator} yields
\begin{equation}
\label{almostdone}
 \partial_{p_\ell}{}^\# (\partial_{q^m}{}^\#)^j =  (\partial_{q^m}{}^\#)^j\partial_{p_\ell}{}^\#  - \delta^\ell{}_m \ps \frac{j}{\hbar} \ps (\partial_{q^m}{}^\#)^{j-1} \eta .
\end{equation}
Substituting \eqref{almostdone} into  the second factor of the last term in \eqref{hairy eq} 
we eventually obtain
\begin{align*}\label{comm}
\partial_{p_\ell}{}^\#( \partial_{q^{1}}{}^{\#})^ {j_1} \cdots  & ( \partial_{q^{n}}{}^{\#})^ {j_n}[\Psi]  \\  & = 
( \partial_{q^{1}}{}^{\#})^ {j_1} \cdots ( \partial_{q^{n}}{}^{\#})^ {j_n}\partial_{p_\ell}{}^\#[\Psi] \\
& \mbox{\ \ \ } - \frac{1}{\hbar}\sum_{m = 1}^n \delta^\ell{}_m\ps j_m ( \partial_{q^{1}}{}^{\#})^ {j_1} \cdots ( \partial_{q^m}{}^{\#})^ {j_m -1}\cdots ( \partial_{q^{n}}{}^{\#})^ {j_n}\eta[\Psi].
\end{align*}
Observe that the first term on the r.h.s.\;here vanishes as $\Psi$ is $J$-polarized. Recalling the  $U(1)$-equivariance of $\Psi$, so that $\eta[\Psi] = \ii\Psi$, the expression above reduces to
\begin{align*}
\partial_{p_\ell}{}^\#( \partial_{q^{1}}{}^{\#})^ {j_1} \cdots  
& ( \partial_{q^n}{}^\#)^ {j_n}[\Psi]  = - \frac{\ii}{\hbar} \ps j_\ell ( \partial_{q^1}{}^{\#})^ {j_1} \cdots 
( \partial_{q^\ell}{}^\#)^{j_\ell -1}\cdots ( \partial_{q^{n}}{}^{\#})^ {j_n}[\Psi].
\end{align*}

Thus the formula for $\partial_{p_\ell}{}^\#[H \bnu \Psi] $  becomes
\begin{align*}
\sum_{k\in \mathbb N} & \sum_{j_1 + \cdots + j_n=k}  \left(\frac{\hbar}{\ii}\right)^k \frac{1}{k!}   
\left(\begin{array}{cc}
k \\
j_1 \cdots j_n
\end{array}
\right) \nonumber \\[2ex]
& \times \Big( \partial_{p_\ell} (\partial_{p_{1}})^{j_1} \cdots (\partial_{p_{n}})^{j_n} [H] 
 \times ( \partial_{q^{1}}{}^{\#})^ {j_1} \cdots ( \partial_{q^{n}}{}^{\#})^ {j_n}[\Psi] \nonumber \\[2ex]
&\ \ \  - \frac{\ii}{\hbar} \ps j_\ell\ps (\partial_{p_{1}})^{j_1} \cdots (\partial_{p_{n}})^{j_n} [H]   \times ( \partial_{q^{1}}{}^{\#})^ {j_1} \cdots ( \partial_{q^{\ell}}{}^{\#})^ {j_\ell -1}\cdots ( \partial_{q^{n}}{}^{\#})^ {j_n}[\Psi]\Big).
\end{align*}
Finally, rewriting 
$$(\partial_{p_1})^{j_1} \cdots  (\partial_{p_n})^{j_n}= (\partial_{p_\ell}) (\partial_{p_{1}})^{j_1} \cdots (\partial_{p_{\ell}})^{j_\ell -1}\cdots (\partial_{p_{n}})^{j_n}$$
in the second term of the sum, reindexing $j_\ell \rightsquigarrow j_\ell +1$ therein  and observing that
$$\frac{j_\ell +1}{k}
\left(\begin{array}{cc}
k \\
j_1 \cdots j_\ell + 1 \cdots j_n
\end{array}
\right) =
\left(\begin{array}{cc}
k-1 \\
j_1 \cdots j_\ell \cdots j_n
\end{array}
\right),
$$
the two terms are seen to cancel and we are done.
\end{proof}

Thus we have proven the existence of a quantum product $\bnu$ on a (prequantization of a) flat bi-polarized manifold $X$. Interestingly Proposition \ref{polprop}  is no longer valid if we use the anti-normal tensor $\mu  =  - \partial_{q^k} \otimes \partial_{p_k}$ as the driver of $\star$, cf. \S\ref{an}.

 Proposition \ref{polprop} also holds when the quantum product is that associated to the Poisson bivector itself, as we now prove. Let us denote this `Moyal quantum product' by $\bpi$.

\begin{cor}
If $H$ is an observable and $\Psi$ is polarized, then $H \bpi \Psi$ is also polarized. 
\end{cor}
\begin{proof}
First of all, expand $\pi = \sum_{k=1}^n \pi_k$, where
\begin{equation}
\pi_k  =  \partial_{p_k} \otimes \partial_{q^k}  -  \partial_{q^k} \otimes \partial_{p_k} 
\label{pik}
\end{equation}
(no sum). We routinely verify that the commutator $[\pi_j^\#,\pi_k^\#] = 0$,\footnote{\, By this we mean that $\pi_j^\#(\pi_k^\#(f\otimes g)) - \pi_k^\#(\pi_j^\#(f\otimes g)) = 0$ on $C^\infty(Y,\mathbb C) \otimes C^\infty(Y,\mathbb C)$.} whence
$$\exp \big( (\hbar/\ii)\pi^\#\big) = \exp \big( (\hbar/\ii)\pi_1^\#\big) \circ \cdots \circ \exp \big( (\hbar/\ii)\pi_n^\#\big).
$$

Next, write \eqref{pik} as $\pi_k = \nu_k + \mu_k$. Then for $H$ an observable and $\Psi$ a polarized wave function we compute
\begin{equation*}
[\nu_k^\#,\mu_k^\#] (H \otimes \Psi)=  - \frac{\ii}{\hbar}\ps \partial_{q^k}  \partial_{p_k} [H] \otimes \Psi
\end{equation*}
and furthermore 
\begin{equation*}
[\nu_k^\#,[\nu_k^\#,\mu_k^\#]] (H \otimes \Psi)= 0 
= [\mu_k^\#,[\nu_k^\#,\mu_k^\#]] (H \otimes \Psi).
\end{equation*}
Invoke the Baker-Campbell-Hausdorff-Zassenhaus formula 
\begin{align*}
\exp(A+B)=  \exp(A)\,& \circ\,\exp(B)\,  \circ\,\exp \big(\! - \!(1/2)[A,B] \big)\\
& \circ \, \exp \big((1/6)[A,[A,B]]+(1/3)[B,[A,B]] \big)\,\circ\,\cdots
\end{align*}
with $A = (\hbar/\ii)\nu_k^\#$ and $B = (\hbar/\ii)\mu_k^\#$ 
and apply it to $H \otimes \Psi$. As the double commutators vanish and as $\exp \ns  \big((\hbar/\ii)\mu_k^\# \big)(H\otimes \Psi) = H \otimes \Psi$ we get 
\begin{align*}
\exp \ns \big((\hbar/\ii) \pi_k^\# \big) & (H\otimes\Psi ) \\
&=
\exp \ns \big((\hbar/\ii) \nu_k^\# \big) \circ\,
\exp \ns\big(\! - \! (\ii\hbar/2)\ps  \partial_{q^k}  \partial_{p_k}\otimes 1\big) (H\otimes\Psi).
\end{align*}
Summing over $k$ this yields
\begin{align*}
\exp \ns \big((\hbar/\ii) \pi^\# \big) & (H\otimes\Psi ) \\
&=
\exp \ns\left((\hbar/\ii) \nu^\#\right) \circ\,
\exp \ns\left((\ii\hbar/2) \varDelta \otimes 1\right) (H\otimes\Psi)
\end{align*}
where $\varDelta = -\sum_{k=1}^n\partial_{q^k}  \partial_{p_k}$ is the Yano Laplacian. From this we finally obtain
\begin{equation}
\label{Agarwal}
F \bpi \Psi = \big[\! \exp \ns\big((\ii\hbar/2) \varDelta\big)F \big]\bnu \Psi
\end{equation}

According to Proposition \ref{polprop}, the r.h.s. here is a polarized wave function, so that $H \bpi \Psi$ is also a polarized wave function.
\end{proof}

Once we have a quantum product at our disposal, we may define quantum operators by means of \eqref{qop}. The final step is to construct the quantum Hilbert space; 
see \cite{Sn1980} for the details of this construction--- it is not essential for our present purposes.

\section{Examples}

Here we specialize to $T^*\mathbb R = \mathbb R^2$ for simplicity; the generalization to $T^*\mathbb R^{n}$ with $n>1$ is immediate.  We apply our method to generate various quantizations of $(\mathbb R^{2},\omega =dp \wedge dq)$.



\subsection{Normal Ordering Quantization}

We first consider the normal $2$-tensor field $\nu=\partial_{p}\otimes  \partial_{q}$.
As the polarization $J$  on $X$ we choose the vertical one, and for $K$ the horizontal one. Then a $J$-polarized wave function on $Y$ has the form $\Psi = \psi(q) e^{{\rm i}\theta}$ where $\psi \in C^\infty(\mathbb R,\mathbb C)$. Observe that the expression above for $\nu$ is exactly of the  form \eqref{decomp}.

As star product on $X$ we take the normal (or standard) one
so that the corresponding  prequantum product is
\begin{equation}
\label{DP}
f \bnu g = m\circ \exp\! \left( \frac{\hbar}{\ii}\ps \nu^\#\right)(f\otimes g).
\end{equation}
In this expression the normal Souriau tensor is 
\begin{equation}
\label{nst}
\nu{}^{\#}=\partial_{p}{}^\# \otimes  \partial_{q}{}^\#
\end{equation}
where the horizontal lifts are $\partial_{p}{}^\# = \partial_{p}$ and $\partial_{q}{}^\# = \partial_{q} - (p/\hbar)\partial_\theta$. Although this expression for $\nu{}^{\#}$ does not have constant coefficients there are no ambiguities in powers of $\nu{}^{\#}{}$ and hence in $\bnu$ as long as we utilize \eqref{nst}, consistent with \eqref{liftoflambda}.

For an observable $F\in \cO$ and a polarized wave function $\Psi$ on $Y$, we thus have 
\begin{equation*}
\nu^\#(F\otimes\Psi)=\partial_p F\otimes  \left(\partial_{q} - \frac{\ii p}{\hbar}\right)\Psi
\label{tnutFtPsi}
\end{equation*}  
(since $\partial_\theta \Psi = \ii \Psi$) and so
we end up with the simple formula
\begin{equation}
F \bnu \Psi = \sum_{k\in\bN}\left(\frac{\hbar}{\ii}\right)^k\frac{1}{k!}\,\partial_p{}^k F \left(\partial_{q} - \frac{\ii p}{\hbar}\right)^k\Psi.
\label{Qter}
\end{equation}

We construct quantum operators $\mathcal Q_\nu F$ according to \eqref{qop}. The series \eqref{Qter} clearly terminates for observables which are polynomial in momenta. In particular, we compute
\begin{equation}
\label{nq}
\mathcal Q_\nu \bigg(\sum_nA_n(q)p^n\bigg)[\Psi] =\sum_n \left(\frac{\hbar}{\ii}\right)^nA_n(q)\psi^{(n)}(q)e^{\ii \theta}
\end{equation}
from which it is apparent that our choice of quantum product in this instance does in fact yield normal-ordering (or standard-ordering) quantization. Notice that (\ref{nq}) applied to $J$-preserving observables, i.e., those of the form $F(p,q)=f(q)p+g(q)$, is consistent with geometric quantization theory.  The quantum representation space is $L^2(\mathbb R,dq)$ as always.

In the context of this example, \eqref{qop} resembles equation (13) of \cite{BNW1998}. In this reference it is however necessary to restrict to the zero section of $T^*\bR$, i.e., $p=0$, in order to be in business. Here there is no such restriction.

We emphasize the crucial role played by the `normal form' \eqref{decomp} with  $s= \partial_{p}$ spanning $J$: In \eqref{DP} the normal tensor $\nu$ {\em must}  be written in the form 
$\nu = \partial_p \otimes \partial_q.$
If instead we use $- \partial_q \otimes \partial_p$ then Proposition \ref{polprop} fails, and we do not obtain a quantization.

\subsection{Example: Anti-normal Ordering Quantization}\label{an}

We interchange the factors in the normal tensor $\nu$, to obtain the `second half' of the Poisson bivector
\begin{equation}
\mu= - \partial_{q}\otimes  \partial_{p},
\label{tmu}
\end{equation}  
\emph{as well as} the polarizations ($J \leftrightarrow  K$). Again \eqref{tmu} is in exactly the normal form \eqref{decomp}, where now $s = - \partial_q$. Then a $K$-polarized wave function has the form $\Phi = \phi(p) e^{\ii (pq/\hbar+\theta)}$ with $\phi \in C^\infty(\mathbb R,\mathbb C).$

The anti-normal quantum product on $Y$ (i.e., \eqref{DP} with $\mu$ in place of~$\nu$) is
\begin{equation*}
F \bmu \Phi = \sum_{k\in\bN}\left(-\frac{\hbar}{\ii}\right)^k\frac{1}{k!}\,\partial_q{}^k F \left[ \left(\partial_{p} +\frac{\ii q}{\hbar}\right)^k\phi \right] e^{\ii (pq / \hbar + \theta)}.
\label{antiQter}
\end{equation*}

The analysis now follows as in the case of normal-ordering quantization, where now the quantum representation space is $L^2(\mathbb R,dp)$. We find for polynomials in $q$ that
\begin{equation}
\label{anq}
\mathcal Q_\mu  \bigg(\sum_nB_n(p)q^n \bigg)[\Phi] =\sum_n \left(-\frac{\hbar}{\ii }\right)^nB_n(p)\phi^{(n)}(p)e^{\ii (pq/\hbar + \theta)}.
\end{equation}
The Fourier transform intertwines the $K$ and $J$ representations; in the former, \eqref{anq} translates into
\begin{equation*}
\mathcal Q_\mu \bigg(\sum_nA_n(q)p^n \bigg)[\Psi] =\sum_n \left(\frac{\hbar}{\ii}\right)^n\frac{d^n}{dq^n}\big(A_n(q)\psi(q)\big)e^{\ii \theta},
\end{equation*}
justifying our terminology.

\subsection{Weyl Quantization}

Our basic setup here is the same as in normal-ordering quantization, except as star product on $\mathbb R^2$ we take the Moyal one with  $\Lambda = \pi$. The quantum product is then $\bpi$ considered previously and \eqref{Agarwal}, expressed in the form
\begin{equation}
\label{agarwal}
\mathcal Q_\pi F = \mathcal Q_\nu\big(\! \exp \ns\big((\ii \hbar/2) \varDelta\big)F\big)
\end{equation}
shows that our version of ``Weyl quantization,'' that is $\mathcal Q_\pi$, is obtained (as expected) from our version of  normal quantization $\mathcal Q_\nu$ via the Agarwal correction (see  \cite{BNW1998,AW1970}). Of course, Weyl quantization produces symmetric operators in contrast to either normal or anti-normal quantization.

Thus far we have  focussed essentially on polynomial observables $F$, in which case the series $F\bullet\Psi$ terminates. But we may still apply our technique, at least formally, when this is not so. As an illustration we \emph{directly} quantize the observable $1/p$ on $\dot T^*\mathbb R \approx \mathbb R \times (\mathbb R \setminus \{0\}).$ As $\varDelta(1/p) = 0$, Weyl quantization will coincide with normal quantization;  a calculation using \eqref{Qter} then yields
\begin{equation*}
\cQ_\pi\!\left(\frac{1}{p}\right)[\Psi] = \frac{1}{p}\sum_{k=0}^\infty p^{-k}\left(\ii\hbar \ps \frac{d}{dq} +p\right)^k [\Psi].
\end{equation*}
 This expression, if convergent, should be a $J$-polarized wave function,\footnote{\  Note that Proposition \ref{polprop} can be applied only if we know that $(1/p) \bnu \Psi$ converges. } and so must be independent of $p$. The only way this can be the case is if $p = -\ii \hbar\ps {d}/{d q}$ in the sense of the functional calculus, so that 
 $$\cQ_\pi\left(\frac{1}{p}\right)[\Psi] = \left(-\ii \hbar \frac{d}{dq}\right)^{-1}\Psi = \frac{\ii}{\hbar}\left( \int_0^q \psi(t)\, dt \right)e^{\ii \theta}.$$
This formal calculation coincides with one based on the integral formula for the Moyal star product; see equation (2.8) in \cite{DS2002}.

\subsection{The Bargmann-Fock Representation}

Start now with $X = \bC$ so that so that the symplectic form reads
\begin{equation*}
\omega=\frac{1}{2\ii}d\bz\wedge{}dz
\end{equation*}
whence 
\begin{equation*}
\pi=2\ii\ps \partial_\bz\wedge\partial_z.
\end{equation*}
The \textit{Bargmann normal tensor}---we should rather say the \textit{Wick tensor}---is defined according to  the usual procedure $\wedge\rightsquigarrow\otimes$, viz.,
\begin{equation*}
\label{Wick}
\nu = 2\ii\ps \partial_\bz\otimes\partial_z.
\end{equation*}
A natural polarization to consider is the antiholomorphic one $J =  {\rm span}_{\mathbb C} \{\partial_{\bar z}\}$, and we take $K = \bar J$.

We have $Y=\mathbb C\times{}S^1$, endowed with the prequantum $1$-form
\begin{eqnarray*}
\alpha&=&\frac{1}{4\ii\hbar}(\bz{}dz-d\bz{}z)+d\theta
\end{eqnarray*}
Polarized wave functions $\Psi:Y\to\mathbb C$ are thus of the form
\begin{eqnarray*}
\Psi = \psi(z)e^{-|z|^2/(4\hbar)}e^{\ii\theta}
\end{eqnarray*}
where $\psi:\bC\to\bC$ is holomorphic: $\partial_\bz\psi=0$. This is the Bargmann representation, the Hermitian inner product of two such wave functions $\Phi$ and $\Psi$ being given by $\langle\Phi,\Psi\rangle=\int_\bC{\!\overline{\Phi}\Psi\,\omega}$.
In the wake of \S 5.A we readily compute
\begin{equation*}
\cQ_\nu \bigg(\sum_nC_n(z)\left(\frac{\bar z}{2\ii}\right)^n \bigg)[\Psi] =\sum_n \left(\frac{\hbar}{\ii }\right)^nC_n(z)\psi^{(n)}(z)e^{-|z|^2/(4\hbar)}e^{\ii\theta}.
\end{equation*}
We can, hence, confirm that our quantization procedure also applies in the K\"ahler case.

Dealing with the Poisson bivector $\pi$, we get accordingly the ``Bargmann'' quantization mapping $\cQ_\pi$ which retains the relationship \eqref{agarwal} with $\cQ_\nu$, where now $\varDelta = -2\ii \partial_{\bar z}\partial_z$ is the K\"ahlerian Laplacian on $\mathbb C$. This formula enables us to find, e.g., 
$$\cQ_\pi (|z|^2)[\Psi] = 2\hbar \bigg(\!z\psi'(z) + \frac{1}{2}\psi(z)\! \bigg)e^{-|z|^2/(4\hbar)}e^{\ii\theta}$$
as expected.

\section{Conclusions}

Our approach is a hybrid of the geometric quantization (GQ) and deformation quantization (DQ) procedures.\footnote{\, As such it is reminiscent of that of \cite{GW2009}.} One the one hand, the GQ of $(X,\omega)$ is unsatisfactory as it quantizes too few observables, viz., those whose Hamiltonian flows preserve the chosen polarization; this point has recently been emphasized in \cite{No2008}. Our technique does not impose this restriction, and consequently we needn't resort to pairing techniques 
to quantize `complicated' observables. On the other, the DQ of $(X,\omega)$ has long been recognized to require additional elements such as a prequantization $(Y,\alpha)$ and/or a polarization. Indeed, the well-known fact that all symplectic (and even Poisson) manifolds admit DQs indicates that DQ cannot be a `genuine' quantization; see also  \cite{Fr1978} and \cite{Hu1980}. There are also related issues with the construction of the quantum state space (\cite{Wa2005}, \cite{No2009}). With our technique, however, this  proceeds as just as in GQ, since polarized wave functions  are already present in $C^\infty(Y,\mathbb C)$. So here we have combined the best of both; in the process we obtain a quantization of, e.g., all polynomials on $\mathbb R^{2n}$. Of course, such a quantization does not satisfy Dirac's ``Poisson bracket $\rightsquigarrow $ commutator'' rule, but this cannot be helped \cite{Go1999}.  In its place, however,  \eqref{convolve} shows that the ``star-commutator $\rightsquigarrow $ commutator'' rule does hold. 

Our  goal in this paper was to show that a prequantization, a polarization and a quantum product can lead to a viable quantization scheme. We have accomplished this under certain circumstances. Our assumption that the phase space is a flat bi-polarized manifold is quite restrictive, even though it covers the physically most important case of Euclidean space. It also covers symplectic tori, which we hope to study next.  As well our construction of a quantum product depended crucially on the exponential nature of the star-products considered and the normal form of Theorem 4.1. More varied and  interesting examples (e.g., cotangent bundles and K\"ahler manifolds) will likely require entirely different techniques of constructing quantum products.

It might also be interesting to extend our construction of a quantum product to the symplectization of the contact manifold $(Y,\alpha)$ in the spirit of what \cite{Ko2005} did for prequantization.

\section{Acknowledgements}

We are greatly indebted to J. E. Andersen, M. Bordemann, R. L. Fernandes, N. P. Landsman, V. Ovsienko and S. Waldmann for enlightening discussions. We also thank T. Masson for helping us with the \LaTeX \ code.

\end{document}